\documentclass[pdflatex,sn-mathphys-ay]{sn-jnl}

\usepackage{graphicx}%
\usepackage{multirow}%
\usepackage{amsmath,amssymb,amsfonts}%
\usepackage{amsthm}%
\usepackage{mathrsfs}%
\usepackage[title]{appendix}%
\usepackage{xcolor}%
\usepackage{textcomp}%
\usepackage{manyfoot}%
\usepackage{booktabs}%
\usepackage{algorithm}%
\usepackage{algorithmicx}%
\usepackage{algpseudocode}%
\usepackage{listings}%
\usepackage{adjustbox}
\usepackage{tikz}
\usetikzlibrary{arrows.meta,calc,angles,quotes,positioning}
\definecolor{inkgray}{RGB}{90,90,90}
\definecolor{primary}{RGB}{31,78,140}
\definecolor{secondary}{RGB}{176,58,46}
\definecolor{fillcap}{RGB}{224,236,247}
\definecolor{boundary}{RGB}{120,120,120}

\theoremstyle{thmstyleone}%
\newtheorem{theorem}{Theorem}%
\newtheorem{proposition}[theorem]{Proposition}%

\theoremstyle{thmstyletwo}%

\theoremstyle{thmstylethree}%
\newtheorem{definition}{Definition}%

\begin{document}

\title[CORAL for Fidelity-Constrained Decorrelation]{CORAL: Constrained Oblique Rotation with Anchored Loadings for Fidelity-Constrained Decorrelation}

\author*[1]{\fnm{Lawrence V.} \sur{Fulton}}\email{lawrence.fulton@bc.edu}
\author[2]{\fnm{Christopher P.} \sur{Fulton}}
\author[1]{\fnm{Arvind} \sur{Sharma}}
\author[1]{\fnm{Aleksandar} \sur{Tomi\'{c}}}

\affil*[1]{\orgdiv{Woods College of Advancing Studies}, \orgname{Boston College}, \orgaddress{\street{140 Commonwealth Avenue}, \city{Chestnut Hill}, \state{MA}, \postcode{02467}, \country{USA}}}
\affil[2]{\orgdiv{United States Air Force Test Pilot School}, \orgaddress{\street{220 Wolfe Ave}, \city{Edwards AFB}, \state{CA}, \postcode{93524}, \country{USA}}}

\abstract{Decorrelating a multivariate system need not destroy source-variable identity. We introduce Constrained Oblique Rotation with Anchored Loadings (CORAL), which minimizes residual cross-correlation while guaranteeing a declared minimum correlation between each transformed variable and its designated source. For a $p$-variable correlation matrix $R$, we show that every exact decorrelator can be written as $R^{-1/2}Q$ for some orthogonal matrix $Q$, and define $\rho_\star(R)$ as the maximum common source fidelity compatible with exact decorrelation. Constructive lower bounds and rigorous analytical upper bounds tightly bracket $\rho_\star$ at $[0.972,0.976]$, $[0.959,0.962]$, and $[0.949,0.950]$ in simulations with $p=6,18,50$, respectively, compared with PCA's largest achievable minimum correlation between distinct principal components and matched source variables of $0.358$, $0.329$, and $0.172$. Corresponding intervals are $[0.751,0.758]$ for World Development Indicators and $[0.826,0.835]$ for wine chemistry data sets. Thus, loss of source-variable identity is not inherent to exact decorrelation but depends on the decorrelator selected. CORAL uses constrained Riemannian optimization and extends to exact support restrictions.}

\keywords{decorrelation, principal component analysis, whitening, interpretability, manifold optimization, constrained optimization, sparse transformations}

\maketitle

\section{Introduction}\label{sec:introduction}

Correlated variables complicate interpretation, attribution, and downstream statistical inference. Regression coefficients may become unstable under respecification, variance decompositions may depend on term order, and transformed representations can lose a clear correspondence to the substantively meaningful variables from which they were constructed \citep{belsley1980regression,kutner2005applied}. Principal component analysis (PCA) provides the canonical orthogonal representation of a correlation or covariance structure \citep{pearson1901lines,hotelling1933analysis,jolliffe2002principal}. When all components are retained and rescaled to unit transformed variance, PCA also yields an exactly decorrelated system. Its optimization criterion, however, is not designed to preserve a one-to-one relationship between transformed variables and their original coordinates. Principal components are typically dense mixtures whose loadings need not align with the substantive meaning attached to any designated source variable \citep{chipman2005interpretable,jolliffe2003modified}.

Several established methods address related aspects of this problem. Sparse PCA seeks simpler loading structures \citep{zou2006sparse,dabkowski2005spca,witten2009penalized,chen2024newbasis}, while varimax, promax, and other factor rotations pursue loading simplicity or interpretable factor structure \citep{kaiser1958varimax,hendrickson1964promax,browne2001overview}. ZCA and related whitening transformations are closer to the present objective because they combine exact decorrelation with a global notion of proximity to the original coordinates \citep{kessy2018optimal,gillard2023polynomial,bell1997independent}. Independent component analysis instead targets statistical independence rather than preservation of designated source variables \citep{hyvarinen2000independent,comon1994independent}. None of these approaches directly asks the following question: how much can a system be decorrelated while guaranteeing, variable by variable, that each transformed coordinate remains correlated with a prespecified source coordinate by at least a declared amount?

We introduce \emph{Constrained Oblique Rotation with Anchored Loadings} (CORAL) to answer that question. Let $X$ be a standardized data matrix with variables $X_1,\ldots,X_p$, and let $T\in\mathbb{R}^{p\times p}$ be the transformation matrix. CORAL defines $\widetilde X=XT$ while requiring $\operatorname{Cor}(\widetilde X_j,X_j)\geq\rho_{\min}$ for every $j$. Thus, $\rho_{\min}$ is the minimum correlation retained between each transformed variable and its designated source. CORAL minimizes residual cross-correlation subject to this fidelity requirement and unit transformed variance. A dual formulation instead fixes an acceptable residual-correlation budget and maximizes the common fidelity floor. The framework combines constrained multi-objective optimization, with roots in goal programming \citep{charnes1957management,ignizio1976goal,romero1991handbook,tamiz1998goal}, and Riemannian optimization \citep{edelman1998geometry,absil2008optimization,boumal2023intromanifolds,townsend2016pymanopt}.

The central feature of CORAL is not merely that the transformed variables are interpretable, but that their relationship to the original coordinates is quantitatively preserved. Specifically, CORAL seeks a transformed system $\widetilde X=XT$ in which cross-correlations among the $\widetilde X_j$ are eliminated or reduced while each transformed coordinate remains correlated with its designated source variable, $\operatorname{Cor}(\widetilde X_j,X_j)\geq\rho_{\min}$. The results show that these objectives are not inherently antagonistic. In the simulated systems, every transformed variable can retain at least $0.949$ correlation with its designated source while the transformed system remains exactly decorrelated. Thus, substantial loss of source-coordinate identity is not a necessary consequence of decorrelation.

The principal theoretical contribution is a characterization of the entire family of exact decorrelators. If $R$ is positive definite, every square transformation satisfying $T^\top RT=I$ is $T=R^{-1/2}Q$ for an orthogonal matrix $Q$. This representation shows that PCA is one exact decorrelator among a continuum of possibilities and makes it possible to define $\rho_\star(R)$, the maximum common source fidelity attainable while preserving exact decorrelation. We derive computable lower and upper bounds for $\rho_\star(R)$, including a subset nuclear-norm bound that can be substantially sharper than the full-matrix trace bound. The empirical frontiers can therefore be interpreted relative to a theoretically defined feasibility boundary rather than only by observing when a numerical solver begins to leave residual correlation.

The numerical analysis evaluated simulated correlation matrices at $p=6$, $p=18$, and $p=50$, eight World Development Indicators (WDI) for 137 countries \citep{worldbank2026wdi}, and the 13-variable University of California, Irvine (UCI) Wine dataset with 178 observations \citep{aeberhard1992wine}. The $p=50$ case provides a dense scalability check. PCA is benchmarked using the most favorable sign-invariant one-to-one assignment for each fidelity measure: minimum matched absolute correlation and mean matched absolute correlation. A separate $p=18$ support analysis fixes prohibited transformation coefficients at zero and distinguishes CORAL's aggregate squared-correlation objective from a min--max objective that minimizes the worst residual pair.

\section{CORAL formulation and solution}\label{sec:formulation}

\subsection{Setup and source fidelity}\label{sec:setup}

Let $X\in\mathbb{R}^{n\times p}$ be a standardized data matrix with symmetric positive-definite sample correlation matrix $R\in\mathbb{R}^{p\times p}$. Let $T=[t_1,\ldots,t_p]\in\mathbb{R}^{p\times p}$ be the transformation matrix and define $\widetilde X=XT$. Then
\begin{equation}
\operatorname{Var}(\widetilde X_j)=t_j^\top Rt_j,\qquad
\operatorname{Cov}(\widetilde X_i,\widetilde X_j)=t_i^\top Rt_j.
\label{eq:variance-covariance}
\end{equation}
CORAL imposes unit transformed variance,
\begin{equation}
t_j^\top Rt_j=1,\qquad j=1,\ldots,p,
\label{eq:unitvariance}
\end{equation}
so $T^\top RT$ is the correlation matrix of the transformed variables. Because both $X_j$ and $\widetilde X_j$ have unit variance,
\begin{equation}
\operatorname{Cor}(\widetilde X_j,X_j)=e_j^\top Rt_j,
\label{eq:fidelity}
\end{equation}
where $e_j$ is the $j$th standard basis vector. The source-fidelity constraint is therefore
\begin{equation}
e_j^\top Rt_j\geq\rho_{\min},\qquad j=1,\ldots,p,
\label{eq:fidelity-floor}
\end{equation}
for a declared $\rho_{\min}\in[0,1]$.

Classical goal programming motivates minimizing total absolute off-diagonal correlation by introducing positive and negative deviation variables \citep{charnes1957management,ignizio1976goal}. The implemented CORAL solver instead uses the smooth squared-correlation criterion
\begin{equation}
D_2(T)=\sum_{i<j}(t_i^\top Rt_j)^2
=\frac{1}{2}\left(\lVert T^\top RT\rVert_F^2-p\right),
\label{eq:squared-decorrelation}
\end{equation}
subject to \eqref{eq:unitvariance} and \eqref{eq:fidelity-floor}. The absolute and squared criteria share the same zero-residual solutions, although they can rank nonzero-residual solutions differently. All numerical frontiers reported below correspond to \eqref{eq:squared-decorrelation}.

\subsection{Oblique-manifold reparameterization}
\label{sec:oblique}

The primal CORAL problem combines the squared decorrelation objective with unit-variance and source-fidelity constraints:
\begin{equation}
\begin{aligned}
\min_T\quad
& D_2(T)
=
\sum_{i<j}(t_i^\top Rt_j)^2\\
\text{subject to}\quad
& t_j^\top Rt_j=1,\qquad j=1,\ldots,p,\\
& e_j^\top Rt_j\geq\rho_{\min},\qquad j=1,\ldots,p.
\end{aligned}
\label{eq:primal}
\end{equation}

To enforce unit transformed variance directly, let $R^{1/2}$ denote the
symmetric positive-definite square root of $R$ and define
\begin{equation}
W=R^{1/2}T,\qquad T=R^{-1/2}W.
\label{eq:reparameterization}
\end{equation}
Because
\[
t_j^\top Rt_j=w_j^\top w_j,
\]
the unit-variance constraints become unit-norm constraints on the columns
of $W$. Thus,
\begin{equation}
W\in\mathcal{OB}(p,p)
=
\left\{
W\in\mathbb{R}^{p\times p}:
\operatorname{diag}(W^\top W)=\mathbf{1}
\right\},
\label{eq:oblique}
\end{equation}
where $\mathcal{OB}(p,p)$ is the oblique manifold of $p\times p$
matrices with unit-norm columns.

The same reparameterization gives
\[
T^\top RT=W^\top W
\]
and
\[
e_j^\top Rt_j=e_j^\top R^{1/2}w_j.
\]
Therefore, \eqref{eq:primal} is equivalent to the computational problem
\begin{equation}
\begin{aligned}
\min_{W\in\mathcal{OB}(p,p)}\quad
& \sum_{i<j}(w_i^\top w_j)^2\\
\text{subject to}\quad
& e_j^\top R^{1/2}w_j\geq\rho_{\min},
\qquad j=1,\ldots,p.
\end{aligned}
\label{eq:primal-oblique}
\end{equation}
Optimization on the oblique manifold therefore enforces unit transformed
variance exactly, leaving only the source-fidelity inequalities to be
enforced numerically. Figure~\ref{fig:coral-geometry} illustrates the CORAL math programming problem.

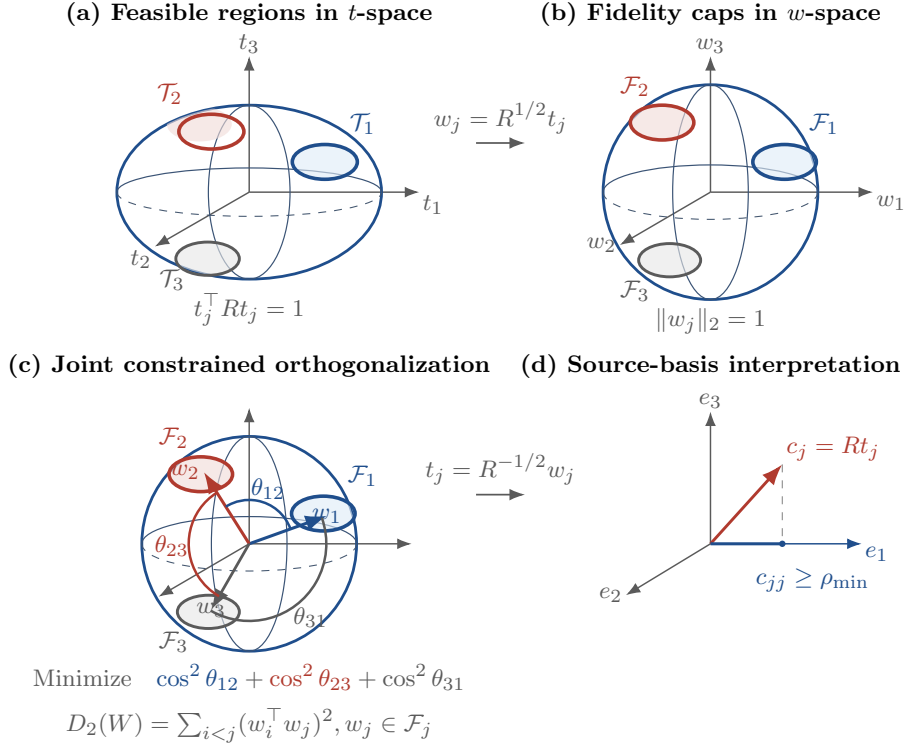
\begin{figure*}[t]
\centering
\begin{tikzpicture}[
    >=Latex,
    every node/.style={font=\small},
    axis/.style={-Latex, inkgray, line width=0.55pt},
    surf/.style={primary, line width=1.05pt},
    contour/.style={primary!55!inkgray, line width=0.45pt},
    backcontour/.style={primary!40!inkgray, dashed, line width=0.40pt},
    feas/.style={primary, line width=1.35pt},
    feas2/.style={secondary, line width=1.35pt},
    feas3/.style={inkgray, line width=1.10pt},
    sel1/.style={primary, -Latex, line width=1.15pt},
    sel2/.style={secondary, -Latex, line width=1.15pt},
    sel3/.style={inkgray, -Latex, line width=1.05pt},
    ghost/.style={inkgray!65, dashed, line width=0.45pt}
]

\begin{scope}[shift={(0,0)}]

  \node[font=\bfseries\small]
    at (0,2.35) {(a) Feasible regions in $t$-space};

  \draw[axis] (0,0) -- (2.25,0)
    node[below right, inner sep=1pt] {$t_1$};
  \draw[axis] (0,0) -- (-1.25,-0.72)
    node[below left, inner sep=1pt] {$t_2$};
  \draw[axis] (0,0) -- (0,1.80)
    node[above, inner sep=1pt] {$t_3$};

  \draw[surf]
    (0,0) ellipse[x radius=1.75,y radius=1.15];

  \draw[backcontour]
    (-1.75,0)
    arc[start angle=180,end angle=360,
        x radius=1.75,y radius=0.32];

  \draw[contour]
    (1.75,0)
    arc[start angle=0,end angle=180,
        x radius=1.75,y radius=0.32];

  \draw[contour]
    (0,0) ellipse[x radius=0.54,y radius=1.15];

  \begin{scope}
    \clip (0,0) ellipse[x radius=1.75,y radius=1.15];

    \fill[fillcap, opacity=0.62]
      (1.00,0.40)
      ellipse[x radius=0.43,y radius=0.23];

    \fill[secondary!18, opacity=0.62]
      (-0.66,0.90)
      ellipse[x radius=0.43,y radius=0.23];

    \fill[inkgray!14, opacity=0.62]
      (-0.55,-0.88)
      ellipse[x radius=0.42,y radius=0.22];
  \end{scope}

  \draw[feas]
    (1.00,0.40)
    ellipse[x radius=0.43,y radius=0.23];

  \draw[feas2]
    (-0.5,0.80)
    ellipse[x radius=0.43,y radius=0.23];

  \draw[feas3]
    (-0.55,-0.88)
    ellipse[x radius=0.42,y radius=0.22];

  \node[primary, font=\scriptsize]
    at (1.5,.9) {$\mathcal T_1$};

  \node[secondary, font=\scriptsize]
    at (-1.05,1.3) {$\mathcal T_2$};

  \node[inkgray, font=\scriptsize]
    at (-1.02,-1.16) {$\mathcal T_3$};

  \node[inkgray, font=\scriptsize]
    at (0,-1.55)
    {$t_j^\top Rt_j=1$};

\end{scope}

\draw[-Latex, inkgray, line width=0.65pt]
  (3,0.65) -- (3.65,0.65)
  node[midway, above, font=\scriptsize]
  {$w_j=R^{1/2}t_j$};

\begin{scope}[shift={(6.1cm,0)}]

  \node[font=\bfseries\small]
    at (0,2.35) {(b) Fidelity caps in $w$-space};

  \draw[axis] (0,0) -- (2.15,0)
    node[below right, inner sep=1pt] {$w_1$};
  \draw[axis] (0,0) -- (-1.20,-0.70)
    node[left, inner sep=1pt] {$w_2$};
  \draw[axis] (0,0) -- (0,1.80)
    node[above, inner sep=1pt] {$w_3$};

  \draw[surf] (0,0) circle[radius=1.42];

  \draw[backcontour]
    (-1.42,0)
    arc[start angle=180,end angle=360,
        x radius=1.42,y radius=0.36];

  \draw[contour]
    (1.42,0)
    arc[start angle=0,end angle=180,
        x radius=1.42,y radius=0.36];

  \draw[contour]
    (0,0) ellipse[x radius=0.50,y radius=1.42];

  \begin{scope}
    \clip (0,0) circle[radius=1.42];

    \fill[fillcap, opacity=0.62]
      (0.98,0.40)
      ellipse[x radius=0.42,y radius=0.23];

    \fill[secondary!18, opacity=0.62]
      (-0.64,0.92)
      ellipse[x radius=0.42,y radius=0.23];

    \fill[inkgray!14, opacity=0.62]
      (-0.54,-0.90)
      ellipse[x radius=0.41,y radius=0.22];
  \end{scope}

  \draw[feas]
    (0.98,0.40)
    ellipse[x radius=0.42,y radius=0.23];

  \draw[feas2]
    (-0.64,0.92)
    ellipse[x radius=0.42,y radius=0.23];

  \draw[feas3]
    (-0.54,-0.90)
    ellipse[x radius=0.41,y radius=0.22];

  \node[primary, font=\scriptsize]
    at (1.5,0.9) {$\mathcal F_1$};

  \node[secondary, font=\scriptsize]
    at (-1.0,1.40) {$\mathcal F_2$};

  \node[inkgray, font=\scriptsize]
    at (-1.00,-1.3) {$\mathcal F_3$};

  \node[inkgray, font=\scriptsize]
    at (0,-1.7)
    {$\|w_j\|_2=1$};

\end{scope}

\begin{scope}[shift={(0,-4.65cm)}]

  \node[font=\bfseries\small]
    at (0,2.35) {(c) Joint constrained orthogonalization};

  \draw[axis] (0,0) -- (2.15,0);
  \draw[axis] (0,0) -- (-1.20,-0.70);
  \draw[axis] (0,0) -- (0,1.80);

  \draw[surf] (0,0) circle[radius=1.42];

  \draw[backcontour]
    (-1.42,0)
    arc[start angle=180,end angle=360,
        x radius=1.42,y radius=0.36];

  \draw[contour]
    (1.42,0)
    arc[start angle=0,end angle=180,
        x radius=1.42,y radius=0.36];

  \draw[contour]
    (0,0) ellipse[x radius=0.50,y radius=1.42];

  \begin{scope}
    \clip (0,0) circle[radius=1.42];

    \fill[fillcap, opacity=0.62]
      (0.98,0.40)
      ellipse[x radius=0.42,y radius=0.23];

    \fill[secondary!18, opacity=0.62]
      (-0.64,0.92)
      ellipse[x radius=0.42,y radius=0.23];

    \fill[inkgray!14, opacity=0.62]
      (-0.54,-0.90)
      ellipse[x radius=0.41,y radius=0.22];
  \end{scope}

  \draw[feas]
    (0.98,0.40)
    ellipse[x radius=0.42,y radius=0.23];

  \draw[feas2]
    (-0.64,0.92)
    ellipse[x radius=0.42,y radius=0.23];

  \draw[feas3]
    (-0.54,-0.90)
    ellipse[x radius=0.41,y radius=0.22];

  \draw[sel1]
    (0,0) -- (1.03,0.37)
    node[inner sep=1pt] {$w_1$};

  \draw[sel2]
    (0,0) -- (-0.60,0.95)
    node[left, inner sep=1pt] {$w_2$};

  \draw[sel3]
    (0,0) -- (-0.50,-0.87)
    node[inner sep=1pt] {$w_3$};

  \coordinate (O) at (0,0);
  \coordinate (Wone) at (1.03,0.37);
  \coordinate (Wtwo) at (-0.60,0.95);
  \coordinate (Wthree) at (-0.50,-0.87);

  \pic[
    draw=primary,
    line width=0.95pt,
    angle radius=5.8mm,
    angle eccentricity=1.32,
    "$\theta_{12}$",
    text=primary,
    font=\scriptsize
  ] {angle=Wone--O--Wtwo};

  \pic[
    draw=secondary,
    line width=0.95pt,
    angle radius=8.0mm,
    angle eccentricity=1.28,
    "$\theta_{23}$",
    text=secondary,
    font=\scriptsize
  ] {angle=Wtwo--O--Wthree};

  \pic[
    draw=inkgray,
    line width=0.95pt,
    angle radius=10.2mm,
    angle eccentricity=1.24,
    "$\theta_{31}$",
    text=inkgray,
    font=\scriptsize
  ] {angle=Wthree--O--Wone};

  \node[primary, font=\scriptsize]
    at (1.5,0.9) {$\mathcal F_1$};

  \node[secondary, font=\scriptsize]
    at (-1.0,1.40) {$\mathcal F_2$};

  \node[inkgray, font=\scriptsize]
    at (-1.00,-1.3) {$\mathcal F_3$};

  \node[inkgray, font=\scriptsize, align=center]
    at (0,-1.76)
    {$\text{Minimize}\quad
      {\color{primary}\cos^2\theta_{12}}
      +{\color{secondary}\cos^2\theta_{23}}
      +{\color{inkgray}\cos^2\theta_{31}}$};

  \node[inkgray, font=\scriptsize, align=center]
    at (0,-2.4)
    {$D_2(W)=\sum_{i<j}(w_i^\top w_j)^2, w_j\in\mathcal F_j$};

\end{scope}

\draw[-Latex, inkgray, line width=0.65pt]
  (3,-4.00) -- (3.65,-4.00)
  node[midway, above, font=\scriptsize]
  {$t_j=R^{-1/2}w_j$};

\begin{scope}[shift={(6.1cm,-4.65cm)}]

  \node[font=\bfseries\small]
    at (0,2.35) {(d) Source-basis interpretation};

  \coordinate (O) at (0,0);
  \coordinate (Eone) at (2.00,0);
  \coordinate (Etwo) at (-1.12,-0.68);
  \coordinate (Ethree) at (0,1.75);
  \coordinate (C) at (0.95,1.05);
  \coordinate (P) at (0.95,0);

  \draw[axis, primary]
    (O) -- (Eone)
    node[below right, inner sep=1pt] {$e_1$};

  \draw[axis]
    (O) -- (Etwo)
    node[left, inner sep=1pt] {$e_2$};

  \draw[axis]
    (O) -- (Ethree)
    node[above, inner sep=1pt] {$e_3$};

  \draw[secondary, -Latex, line width=1.10pt]
    (O) -- (C)
    node[above right, inner sep=1pt] {$c_j=Rt_j$};

  \draw[ghost] (C) -- (P);
  \fill[primary] (P) circle[radius=1.2pt];
  \draw[primary, line width=1.15pt] (O) -- (P);

  \node[primary, font=\scriptsize]
    at (1.35,-0.48)
    {$c_{jj}\geq\rho_{\min}$};

\end{scope}

\end{tikzpicture}

\caption{Geometric interpretation of CORAL. (a) Each transformation column has a coordinate-specific feasible region $\mathcal T_j$ on the unit-variance ellipsoid. (b) Under $w_j=R^{1/2}t_j$, these regions map to fidelity caps $\mathcal F_j$ on the unit hypersphere. (c) CORAL jointly selects $w_j\in\mathcal F_j$ to minimize residual correlation; exact decorrelation occurs when the caps admit an orthogonal frame. (d) Mapping back gives $c_j=Rt_j$, whose entries are correlations with the original variables, while the designated coordinate satisfies $c_{jj}=e_j^\top Rt_j\geq\rho_{\min}$. Projected contours are schematic.}
\label{fig:coral-geometry}
\end{figure*}

Define
$
g_j(W)=\rho_{\min}-e_j^\top R^{1/2}w_j\leq0.
$
The Augmented Lagrangian (ALM) \citep{rockafellar1974augmented}, used to solve \eqref{eq:primal-oblique} is
\begin{equation}
\mathcal{L}_{\eta}(W,\lambda)
=
\sum_{i<j}(w_i^\top w_j)^2
+
\frac{1}{2\eta}
\sum_{j=1}^p
\left[
\max\{0,\lambda_j+\eta g_j(W)\}^2-\lambda_j^2
\right],
\label{eq:alm}
\end{equation}
with multiplier update
$
\lambda_j
\leftarrow
\max\{0,\lambda_j+\eta g_j(W^\star)\}.
$
The ALM penalty parameter $\eta$ is increased across
outer iterations until the maximum fidelity violation falls below the
declared numerical tolerance. When
$\lambda_j+\eta g_j(W)>0$, the corresponding augmented term reduces to
$\lambda_j g_j(W)+(\eta/2)g_j(W)^2$; the projection through
$\max\{0,\cdot\}$ prevents satisfied inequality constraints from being
penalized unnecessarily. Because the problem is non-convex, it is solved
from multiple initializations and the best feasible solution is retained.

\subsection{Dual decorrelation-budget formulation}
\label{sec:dual}
The dual formulation fixes an acceptable residual-correlation budget and maximizes the common source fidelity supported by that budget. It solves
\begin{equation}
\begin{aligned}
\max_{T,\gamma}\quad & \gamma\\
\text{subject to}\quad
& t_j^\top Rt_j=1,\qquad j=1,\ldots,p,\\
& e_j^\top Rt_j\geq \gamma,\qquad j=1,\ldots,p,\\
& \sum_{i<j}(t_i^\top Rt_j)^2\leq\epsilon.
\end{aligned}
\label{eq:dual}
\end{equation}
Here, $\gamma=\min_j\operatorname{Cor}(\widetilde X_j,X_j)$ is the common source-fidelity floor, and $\epsilon$ is the declared residual-correlation budget. The fidelity and decorrelation inequalities are enforced by ALM on the product of the oblique manifold and $\mathbb{R}$. Varying $\epsilon$ traces the fidelity--decorrelation frontier from the budget-constrained direction.

\section{Exact-decorrelation theory}\label{sec:theory}

\subsection{The family of exact decorrelators}\label{sec:exact-family}

The rotational non-uniqueness of whitening is well established \citep{kessy2018optimal}. Under the present right-multiplication convention, the family of square exact decorrelators can be written as follows.

\begin{proposition}[Characterization of square exact decorrelators]
\label{prop:exact-decorrelators}
Let $R\in\mathbb{R}^{p\times p}$ be symmetric positive definite, define
$\mathcal{D}(R)=\{T\in\mathbb{R}^{p\times p}:T^\top RT=I\}$,
and let $\mathcal{O}(p)$ denote the group of $p\times p$ orthogonal matrices.
Then
\begin{equation}
\mathcal{D}(R)=\{R^{-1/2}Q:Q\in O(p)\},
\label{eq:exact-family}
\end{equation}
\end{proposition}
\begin{proof}
If $T\in\mathcal{D}(R)$, define $Q=R^{1/2}T$. Then
$Q^\top Q=T^\top RT=I$, so $Q\in O(p)$ and
$T=R^{-1/2}Q$. Conversely, if $T=R^{-1/2}Q$ for
$Q\in O(p)$, then
$T^\top RT
=
Q^\top R^{-1/2}RR^{-1/2}Q
=
Q^\top Q
=
I.
$
\end{proof}

Proposition~\ref{prop:exact-decorrelators} makes the non-uniqueness of exact decorrelation explicit. If $R=V\Lambda V^\top$, the $R$-normalized PCA transform $T_{\mathrm{PCA}}=V\Lambda^{-1/2}$ belongs to \eqref{eq:exact-family} and corresponds to the orthogonal choice $Q=V$. ZCA corresponds to another choice. CORAL's dense exact-decorrelation question is therefore not whether exact decorrelation exists, but which member of $R^{-1/2}O(p)$ best preserves source-variable identity.

\subsection{Maximum fidelity compatible with exact decorrelation}\label{sec:rho-star}

\begin{definition}[Exact-decorrelation fidelity threshold]
\label{def:rho-star}
The maximum common source fidelity compatible with exact decorrelation is
\begin{equation}
\rho_\star(R)
=
\max_{Q\in O(p)}
\min_{1\leq j\leq p}
e_j^\top R^{1/2}q_j,
\label{eq:rho-star}
\end{equation}
where $q_j$ denotes column $j$ of $Q$.
\end{definition}

\begin{proposition}[Feasibility threshold]
\label{prop:rho-star-threshold}
An exactly decorrelating transformation satisfying
$e_j^\top Rt_j\geq\rho_{\min}$ for all $j$ exists if and only if
$\rho_{\min}\leq\rho_\star(R)$. The maximum in
\eqref{eq:rho-star} is attained.
\end{proposition}

\begin{proof}
By Proposition~\ref{prop:exact-decorrelators}, every exact decorrelator
can be written as $T=R^{-1/2}Q$ for some $Q\in O(p)$. Because
$\widetilde X_j=Xt_j$ and both $X_j$ and $\widetilde X_j$ have unit
variance,
$
\operatorname{Cor}(\widetilde X_j,X_j)=e_j^\top Rt_j.
$
Since $t_j=R^{-1/2}q_j$,
$
e_j^\top Rt_j
=
e_j^\top R^{1/2}q_j.
$
Thus, for a fixed $Q$, the largest common fidelity floor it supports is
$
\min_j e_j^\top R^{1/2}q_j.
$
By definition in \eqref{eq:rho-star}, the maximum of this common fidelity over all $Q\in O(p)$ is $\rho_\star(R)$. Therefore an exact decorrelator satisfying the declared
fidelity floor exists if and only if $\rho_{\min}\leq\rho_\star(R)$. Finally, $O(p)$ is compact and $\min_j e_j^\top R^{1/2}q_j$ is continuous in $Q$, so the maximum is
attained.
\end{proof}

The threshold in Proposition~\ref{prop:rho-star-threshold} separates two regimes. Below $\rho_\star(R)$, zero residual correlation and the declared fidelity floor are jointly feasible. Above $\rho_\star(R)$, any feasible CORAL solution must retain nonzero residual correlation because the fidelity-feasible set no longer intersects the exact-decorrelator family.

\subsection{Rigorous bounds for $\rho_\star(R)$}\label{sec:rho-bounds}

\begin{proposition}[Bounds on the exact-decorrelation fidelity threshold]
\label{prop:rho-bounds}
The exact-decorrelation fidelity threshold satisfies
\begin{equation}
\min_j (R^{1/2})_{jj}
\leq
\rho_\star(R)
\leq
\frac{\operatorname{tr}(R^{1/2})}{p}
\leq 1.
\label{eq:rho-basic-bounds}
\end{equation}
A sharper upper bound is available for any nonempty subset
$S\subseteq\{1,\ldots,p\}$:
\begin{equation}
\rho_\star(R)
\leq
\frac{
\left\lVert \bigl(R^{1/2}\bigr)_{[:,S]}
\right\rVert_\ast
}{|S|},
\label{eq:rho-subset-bound}
\end{equation}
where $\bigl(R^{1/2}\bigr)_{[:,S]}$ contains the columns of $R^{1/2}$
indexed by $S$, and $\lVert\cdot\rVert_\ast$ denotes the nuclear norm,
the sum of the singular values. Therefore,
\begin{equation}
\rho_\star(R)
\leq
\min_{\varnothing\neq S\subseteq\{1,\ldots,p\}}
\frac{
\left\lVert \bigl(R^{1/2}\bigr)_{[:,S]}
\right\rVert_\ast
}{|S|}.
\label{eq:rho-best-subset-bound}
\end{equation}
\end{proposition}

\begin{proof}
For the lower bound, choose $Q=I$ in \eqref{eq:rho-star}. The resulting
exact decorrelator has common source fidelity
$
\min_j e_j^\top R^{1/2}e_j
=
\min_j (R^{1/2})_{jj}.
$
Because $\rho_\star(R)$ is the maximum common fidelity over all
$Q\in O(p)$, it cannot be smaller than this attainable value.

For an upper bound, the minimum fidelity cannot exceed the average
fidelity. Thus, for any $Q\in O(p)$,
$
\min_j e_j^\top R^{1/2}q_j
\leq
\frac{1}{p}
\sum_{j=1}^p e_j^\top R^{1/2}q_j
=
\frac{1}{p}\operatorname{tr}(R^{1/2}Q).
$
The maximum of $\operatorname{tr}(R^{1/2}Q)$ over orthogonal $Q$ is
$\lVert R^{1/2}\rVert_\ast$. Since $R^{1/2}$ is positive definite,
its nuclear norm equals its trace. Hence
$
\rho_\star(R)
\leq
\frac{\operatorname{tr}(R^{1/2})}{p}.
$
Because $R$ is a correlation matrix, $\operatorname{tr}(R)=p$.
Concavity of the square root then gives
$
\frac{\operatorname{tr}(R^{1/2})}{p}\leq1.
$

The same argument can be applied to any nonempty subset $S$. For any
$Q\in O(p)$,
$
\min_j e_j^\top R^{1/2}q_j
\leq
\frac{1}{|S|}
\sum_{j\in S}e_j^\top R^{1/2}q_j.
$
Because $R^{1/2}$ is symmetric, the sum on the right is
$
\operatorname{tr}\left[
\bigl(R^{1/2}\bigr)_{[:,S]}^\top Q_{[:,S]}
\right].
$
The columns of $Q_{[:,S]}$ are orthonormal, so this trace cannot exceed
the nuclear norm of $\bigl(R^{1/2}\bigr)_{[:,S]}$. Therefore,
$
\rho_\star(R)
\leq
\frac{
\left\lVert \bigl(R^{1/2}\bigr)_{[:,S]}
\right\rVert_\ast
}{|S|}.
$
Since the bound holds for every nonempty $S$, taking the minimum over
all such subsets gives \eqref{eq:rho-best-subset-bound}.
\end{proof}

Numerically optimizing \eqref{eq:rho-star} produces a feasible orthogonal matrix and therefore a constructive lower bound on the true $\rho_\star(R)$ even when global optimality is not certified. Combining that achieved fidelity with Proposition~\ref{prop:rho-bounds} yields a certified interval. When the interval lies entirely between two adjacent values of the fidelity grid, the observed transition from exact to inexact decorrelation has a theoretical explanation independent of the primal solver.

\subsection{A sign-invariant PCA benchmark}\label{sec:pca-benchmark}

The sign of a PCA eigenvector is arbitrary, so signed component-to-source correlations cannot be used directly as an interpretability benchmark. For $R=V\Lambda V^\top$ and $T_{\mathrm{PCA}}=V\Lambda^{-1/2}$, the matrix of correlations between original standardized variables and unit-variance PCA coordinates is
$
C_{\mathrm{PCA}}=RT_{\mathrm{PCA}}=V\Lambda^{1/2}.
\label{eq:pca-correlation}
$
A fair one-to-one source assignment must therefore depend on
$|C_{\mathrm{PCA}}|$. Let $\Pi_p$ denote the set of all permutations of
$\{1,\ldots,p\}$. Because CORAL's primary guarantee concerns minimum
fidelity, we report the bottleneck assignment
$
\rho_{\mathrm{PCA,min}}
=
\max_{\sigma\in\Pi_p}
\min_j |C_{j,\sigma(j)}|,
\label{eq:pca-bottleneck}
$
which gives PCA the most favorable assignment for its worst-matched
source variable. We separately report the sum-optimal assignment
$
\bar\rho_{\mathrm{PCA}}
=
\frac{1}{p}
\max_{\sigma\in\Pi_p}
\sum_j |C_{j,\sigma(j)}|.
\label{eq:pca-sum}
$

\subsection{Support-restricted exact decorrelation}
\label{sec:support-theory}

Let $\Omega\subseteq\{1,\ldots,p\}^2$ denote the allowed entries of the
transformation matrix $T$, with $(j,j)\in\Omega$ for every $j$. Define
\begin{equation}
\mathcal{C}_{\Omega}(R)
=
\left\{
T:
t_j^\top Rt_j=1\ \forall j,\;
T_{ij}=0\ \text{for }(i,j)\notin\Omega
\right\}.
\label{eq:support-set}
\end{equation}
Two natural support-constrained residual criteria are
\begin{equation}
\delta_{2,\Omega}(R)
=
\min_{T\in\mathcal{C}_{\Omega}(R)}
\frac{1}{2}
\left(
\lVert T^\top RT\rVert_F^2-p
\right),
\label{eq:delta2}
\end{equation}
and
\begin{equation}
\delta_{\infty,\Omega}(R)
=
\min_{T\in\mathcal{C}_{\Omega}(R)}
\max_{i<j}|t_i^\top Rt_j|.
\label{eq:deltainf}
\end{equation}
The first matches CORAL's squared-correlation objective. The second
minimizes the largest residual correlation.

\begin{proposition}[Existence of support-constrained minima]
\label{prop:support-floor}
Assume $R$ is positive definite and every transformation column has at
least one allowed entry. Then $\mathcal{C}_{\Omega}(R)$ is compact and
the minima in \eqref{eq:delta2} and \eqref{eq:deltainf} are attained.
Moreover,
$\delta_{2,\Omega}(R)=0$ if and only if
$\delta_{\infty,\Omega}(R)=0$, which occurs if and only if there exists
$Q\in O(p)$ such that $R^{-1/2}Q$ respects $\Omega$. If no such $Q$
exists, both minima are strictly positive.
\end{proposition}

\begin{proof}
For each column, the unit-$R$ ellipsoid intersected with its allowed
coordinate subspace is closed and bounded, hence compact. Their finite
Cartesian product is therefore compact. Because both residual criteria
are continuous, their minima are attained.

Either minimum is zero exactly when all off-diagonal elements of
$T^\top RT$ are zero. The unit-variance constraints then imply
$T^\top RT=I$. By Proposition~\ref{prop:exact-decorrelators},
$T=R^{-1/2}Q$ for some $Q\in O(p)$. Thus, zero residual is possible if
and only if an exact decorrelator respects $\Omega$. Otherwise, the
attained nonnegative minima must be strictly positive.
\end{proof}

Proposition~\ref{prop:support-floor} establishes when a positive
support-induced floor must exist, but it does not certify its numerical
value for a particular non-convex instance. The support computations
below are therefore reported as multistart best-found values rather
than global minima.

\section{Study Design}\label{sec:methods}

\subsection{Simulation correlation structures}
\label{sec:simulation}

For each simulated dimension, a matrix
$G\in\mathbb{R}^{p\times p}$ was generated with independent standard
normal entries, and the correlation matrix was constructed as
\[
M=\frac{GG^\top}{p}+0.5I,
\qquad
\Delta=\operatorname{diag}(M),
\qquad
R=\Delta^{-1/2}M\Delta^{-1/2}.
\]
This construction produced positive-definite correlation matrices with unit diagonal and moderate off-diagonal correlation. Dense simulations used $p=6$, $p=18$, and $p=50$. For $p=50$, the maximum and mean absolute off-diagonal correlations were $0.2909$ and $0.0748$, respectively. The dimensions $p=6$, $18$, and $50$ were chosen to span small, moderate, and larger systems; $p=18$ also serves as the basis for the support-restriction ablation, while $p=50$ provides a scalability check rather than a controlled dimension-effect experiment.

For the $p=18$ support-restriction ablation, an off-diagonal transformation coefficient was allowed when $|R_{ij}|>0.15$, with all diagonal coefficients retained. This yielded 92 allowed off-diagonal coefficients among 306 possible directed off-diagonal entries, or 30.1\% density. Including the 18 diagonal coefficients gave 110 free transformation coefficients and 214 structural zeros.

The exact-decorrelator family $R^{-1/2}O(18)$ has $\dim O(18)=18(17)/2=153$ rotational degrees of freedom. In the $p=18$ support-restriction ablation, the 214 declared zeros therefore exceeded the rotational degrees of freedom, while the 110 allowed coefficients were fewer than the 171 scalar equations in $T^\top RT=I$. These counts provided a generic dimension diagnostic, not an infeasibility proof for the realized $R$, because special alignment between $R$ and the support can permit sparse exact decorrelators.

\subsection{Numerical optimization}\label{sec:numerical}

Dense CORAL was solved on the oblique manifold using Riemannian trust-region steps \citep{absil2007trust} with ALM fidelity constraints. Every non-convex optimization reported in the study used 100 independent starts, retaining the best feasible solution. The exact-decorrelation threshold \eqref{eq:rho-star} was solved as a max--min problem on the square Stiefel manifold $\mathrm{St}(p,p)=O(p)$ \citep{edelman1998geometry,song2024linear}, with common fidelity as a scalar decision variable and the fidelity inequalities enforced by ALM \citep{rockafellar1974augmented}. Solver tolerances are verified in Appendix~A. Each achieved $\rho_\star$ value provides a constructive lower bound because its corresponding $Q$ is orthogonal to numerical precision; global optimality is not claimed. Upper bounds follow Proposition~\ref{prop:rho-bounds}. The empirical data set analyses enumerated all nonempty subsets to obtain the strongest subset nuclear-norm bound.

The support-restricted solver removes forbidden $T_{ij}$ coefficients from the decision vector, so support is exact by construction rather than approximately enforced by a penalty. Unit-variance equalities and fidelity inequalities are solved by sequential quadratic programming from multiple starts. The same reduced parameterization is used to estimate \eqref{eq:delta2} and \eqref{eq:deltainf}. Because these problems remain non-convex, the resulting values are explicitly labeled best-found.

\subsection{PCA assignment and empirical data}
\label{sec:data}

PCA components were rescaled to unit variance under the same $R$ metric as CORAL. Because eigenvector signs are arbitrary and principal components have no intrinsic one-to-one correspondence with the original variables, assignments were based on absolute source-component correlations. The bottleneck assignment in \eqref{eq:pca-bottleneck} maximized minimum source fidelity, while the sum-optimal assignment in \eqref{eq:pca-sum} maximized mean source fidelity. The latter was solved using the Hungarian algorithm \citep{kuhn1955hungarian}. These assignments provide sign-invariant PCA benchmarks under the fidelity criterion used for CORAL.

The first empirical application used the 2019 World Development Indicators (WDI) \citep{worldbank2026wdi}. After excluding regional and income-group aggregates, the complete-case sample contained 137 countries observed on eight indicators: fertility rate, GDP per capita, government expenditure on education, gross capital formation, internet use, life expectancy at birth, under-5 mortality, and urban population. Adult literacy and energy use per capita were excluded because retaining either would have reduced the complete-case sample to 65 countries. The second application used the UCI Wine dataset \citep{aeberhard1992wine}, containing 178 observations on 13 chemical and visual measurements. Each empirical dataset was standardized before constructing its sample correlation matrix $R$, so the same CORAL, ZCA, and PCA fidelity definitions applied across the simulation and empirical analyses.

\section{Results}
\label{sec:results}

Table~\ref{tab:rho-star} summarizes the exact-decorrelation thresholds, the ZCA baseline, and the sign-invariant PCA benchmarks across all five systems. ZCA corresponds to the feasible choice $Q=I$ and therefore provides a constructive lower bound on $\rho_\star(R)$. The CORAL max--min search improves this lower bound in every system, while the rigorous upper bounds tightly constrain the remaining gap. For the $p=6$ and $p=18$ simulations, the resulting intervals place $\rho_\star(R)$ between the tested fidelity levels $0.95$ and $0.99$. At $p=50$, the interval straddles $0.95$. The empirical thresholds are lower, with $\rho_\star(R)$ near $0.75$ for WDI and $0.83$ for Wine.

\begin{table*}[h]
\caption{Fidelity benchmarks and bounds under exact decorrelation. ZCA reports the minimum source fidelity for the $Q=I$ decorrelator. CORAL reports the best achieved minimum fidelity from the max--min search over $Q\in O(p)$, while CORAL upper gives a rigorous upper bound on $\rho_\star(R)$. The simulations use the trace bound; WDI and Wine use the strongest enumerated subset bound. PCA reports sign-invariant minimum and mean fidelities under bottleneck and sum-optimal assignments, respectively.}
\label{tab:rho-star}
\centering
\footnotesize
\begin{tabular}{@{}lrrrrr@{}}
\toprule
System & $p$ & ZCA & CORAL & CORAL upper & PCA  min / mean\\
\midrule
Synthetic & 6  & 0.9565 & 0.9718 & 0.9757 & 0.3579 / 0.6665\\
Synthetic & 18 & 0.9365 & 0.9585 & 0.9616 & 0.3291 / 0.4838\\
Synthetic & 50 & 0.9215 & 0.9491 & 0.9504 & 0.1716 / 0.3206\\
WDI 2019  & 8  & 0.6706 & 0.7509 & 0.7579 & 0.1797 / 0.5399\\
Wine      & 13 & 0.7107 & 0.8262 & 0.8355 & 0.2482 / 0.4662\\
\botrule
\end{tabular}
\end{table*}

The ZCA results clarify what CORAL gains from optimizing the orthogonal freedom within the exact-decorrelator family. Because ZCA sets $Q=I$, its minimum source fidelity is $\min_j(R^{1/2})_{jj}$ and must satisfy $\rho_{\mathrm{ZCA}}\leq\rho_\star(R).$

The optimized exact decorrelators improve on this baseline in all five systems. The increase is systematic in the simulated systems, from $0.9565$ to $0.9718$ at $p=6$, from $0.9365$ to $0.9585$ at $p=18$, and from $0.9215$ to $0.9491$ at $p=50$. The gain is larger in the empirical applications, from $0.6706$ to $0.7509$ for WDI and from $0.7107$ to $0.8262$ for Wine. Thus, ZCA already preserves substantial source identity while exactly decorrelating, but the choice $Q=I$ need not maximize the weakest source match.

PCA exhibits a different pattern. In the three simulated systems, its best bottleneck fidelities are $0.3579$, $0.3291$, and $0.1716$, compared with CORAL exact-decorrelator fidelities of at least $0.9718$, $0.9585$, and $0.9491$. For WDI and Wine, the corresponding PCA minima are $0.1797$ and $0.2482$, compared with $0.7509$ and $0.8262$. These comparisons are invariant to PCA component signs and ordering. Unlike ZCA, however, PCA is not a lower bound on $\rho_\star(R)$; its assignment criterion is unrelated to the max--min fidelity optimization defining the threshold.

\subsection{Dense simulation results}
\label{sec:dense-results}

Table~\ref{tab:dense-frontiers} combines the hard-constrained primal results for all three simulated dimensions. For $p=6$ and $p=18$, exact decorrelation is retained through $\rho_{\min}=0.95$. At $p=50$, exact decorrelation is retained through $0.85$, while the maximum residual correlation at $0.95$ is only $0.0035$. At $\rho_{\min}=0.99$, exact decorrelation is impossible in all three systems.

\begin{table*}[h]
\caption{Hard-constrained CORAL frontiers for the three dense simulated systems. All reported numerical solutions satisfy the declared fidelity constraint.}
\label{tab:dense-frontiers}
\centering
\scriptsize
\setlength{\tabcolsep}{4pt}
\begin{tabular}{@{}r|rr|rr|rr@{}}
\toprule
& \multicolumn{2}{c}{$p=6$} & \multicolumn{2}{c}{$p=18$} & \multicolumn{2}{c}{$p=50$}\\
\cmidrule(lr){2-3}\cmidrule(lr){4-5}\cmidrule(lr){6-7}
$\rho_{\min}$
& Max $|\mathrm{offdiag}|$ & Min fidelity
& Max $|\mathrm{offdiag}|$ & Min fidelity
& Max $|\mathrm{offdiag}|$ & Min fidelity\\
\midrule
0.30 & 0.0000 & 0.3000 & 0.0000 & 0.3000 & --     & --\\
0.50 & 0.0000 & 0.5000 & 0.0000 & 0.5000 & 0.0000 & 0.5000\\
0.70 & 0.0000 & 0.7000 & 0.0000 & 0.7000 & 0.0000 & 0.7000\\
0.85 & 0.0000 & 0.8500 & 0.0000 & 0.8500 & 0.0000 & 0.8500\\
0.95 & 0.0000 & 0.9500 & 0.0000 & 0.9500 & 0.0035 & 0.9500\\
0.99 & 0.1514 & 0.9899 & 0.2045 & 0.9900 & 0.1700 & 0.9899\\
\bottomrule
\end{tabular}
\end{table*}
The threshold bounds explain these transitions. For $p=6$, $0.9718\leq\rho_\star(R)\leq0.9757$, while for $p=18$, $0.9585\leq\rho_\star(R)\leq0.9616$. Exact decorrelation at $\rho_{\min}=0.95$ is therefore feasible in both systems and impossible at $0.99$. For $p=50$, $0.9491\leq\rho_\star(R)\leq0.9504$, so the bounds do not resolve exact feasibility at precisely $0.95$. The nearly zero primal residual at that target is consistent with its location at the threshold.

The $p=50$ experiment also provides a scalability check. The five-point primal frontier required $5.42$ seconds, while the 100-start $\rho_\star$ search required $130.55$ seconds. The primal operating points therefore remained inexpensive relative to direct max--min estimation of the exact-decorrelation threshold (see Appendix~\ref{app:final-verification} for full numerical diagnostics).

\subsection{Empirical applications}
\label{sec:empirical-results}

The empirical applications exhibit lower exact-decorrelation thresholds than the simulated systems. WDI has maximum absolute pairwise correlation $0.906$ and mean absolute pairwise correlation $0.434$. Wine has maximum absolute pairwise correlation $0.865$ and condition number $45.5$. In both datasets, exact decorrelation is attainable at $\rho_{\min}=0.70$ but not at $0.85$.

For WDI, the strongest subset upper bound is generated by fertility, internet use, life expectancy, and under-5 mortality, giving $0.7509\leq\rho_\star(R_{\mathrm{WDI}})\leq0.7579$. This subset identifies the variables producing the strongest upper bound and should not be interpreted causally. Because $0.85$ exceeds the rigorous upper bound, exact decorrelation at that fidelity level is impossible.

For Wine, the strongest subset consists of total phenols, flavanoids, and OD280/OD315 of diluted wines, giving $0.8262\leq\rho_\star(R_{\mathrm{Wine}})\leq0.8355$. Again, $0.85$ exceeds the rigorous upper bound. Residual correlation at $\rho_{\min}=0.85$ is therefore required by the correlation structure rather than evidence that the optimizer failed to locate an exact decorrelator.

Table~\ref{tab:real-primal} reports the corresponding hard-constrained primal frontiers. WDI incurs the sharper transition: maximum residual correlation rises from zero at $\rho_{\min}=0.70$ to $0.2832$ at $0.85$ and $0.8296$ at $0.99$. Wine rises more gradually at first, reaching $0.0963$ at $0.85$, but reaches $0.7605$ at $0.99$. The declared fidelity floor is met throughout both sweeps.

\begin{table*}[t]
\caption{Hard-constrained primal frontiers for the two empirical datasets.}
\label{tab:real-primal}
\centering
\scriptsize
\begin{tabular}{@{}r|rr|rr@{}}
\toprule
& \multicolumn{2}{c}{WDI 2019 ($p=8$, $n=137$)} & \multicolumn{2}{c}{Wine ($p=13$, $n=178$)}\\
\cmidrule(lr){2-3}\cmidrule(lr){4-5}
$\rho_{\min}$ & Max $|\mathrm{offdiag}|$ & Min fidelity & Max $|\mathrm{offdiag}|$ & Min fidelity\\
\midrule
0.50 & 0.0000 & 0.5000 & 0.0000 & 0.5000\\
0.70 & 0.0000 & 0.7000 & 0.0000 & 0.7000\\
0.85 & 0.2832 & 0.8500 & 0.0963 & 0.8500\\
0.95 & 0.6573 & 0.9500 & 0.5434 & 0.9500\\
0.99 & 0.8296 & 0.9900 & 0.7605 & 0.9900\\
\botrule
\end{tabular}
\end{table*}

The dual formulation traces the same trade-off from the residual-correlation side. Table~\ref{tab:dual-real} reports the maximum common fidelity $\gamma$ under declared squared-correlation budgets $\epsilon$. The budget is active to numerical precision throughout the sweep. For example, $\epsilon=0.10$ supports $\gamma=0.8098$ for WDI and $\gamma=0.8715$ for Wine. Both exceed their exact-decorrelation thresholds because the dual permits nonzero residual correlation.

\begin{table}[t]
\caption{Dual common fidelity under declared squared-correlation budgets.}
\label{tab:dual-real}
\centering
\scriptsize
\begin{tabular}{@{}rrr@{}}
\toprule
$\epsilon$ (squared-correlation budget) & WDI $\gamma$ & Wine $\gamma$\\
\midrule
0.01 & 0.7711 & 0.8428\\
0.05 & 0.7940 & 0.8602\\
0.10 & 0.8098 & 0.8715\\
0.20 & 0.8305 & 0.8858\\
0.40 & 0.8570 & 0.9033\\
\botrule
\end{tabular}
\end{table}

The empirical PCA benchmarks show the same distinction as the simulations. PCA's best minimum source fidelity is $0.1797$ for WDI and $0.2482$ for Wine, compared with exact-decorrelator fidelities of at least $0.7509$ and $0.8262$. The empirical correlation structures therefore reduce the fidelity available under exact decorrelation, but the reduction is not specific to CORAL.

\subsection{Support-restriction ablation}
\label{sec:support-results}

The support-restriction ablation isolates the effect of limiting which source variables may contribute to each transformed coordinate. For the $p=18$ simulation, off-diagonal transformation entries are allowed when $|R_{ij}|>0.15$, with all diagonal entries retained. The resulting support contains 92 of 306 directed off-diagonal entries (30.1\%), giving 110 free transformation coefficients and 214 structural zeros.

The unrestricted exact-decorrelator family $R^{-1/2}O(18)$ has 153 rotational degrees of freedom, while exact decorrelation imposes 171 independent symmetric equations. The 214 structural zeros exceed the rotational degrees of freedom, and the 110 free coefficients are fewer than the number of exact-decorrelation equations. These counts are a generic dimension diagnostic rather than an infeasibility proof because special alignment between $R$ and the support can still permit sparse exact decorrelators.

Table~\ref{tab:support-frontier} reports the exact-support CORAL frontier. Every reported solution satisfies the declared support and unit-variance constraints to numerical tolerance. Unlike the dense case, the best-found solutions retain nonzero residual correlation throughout the fidelity sweep. Maximum residual correlation varies between approximately $0.16$ and $0.20$ through $\rho_{\min}=0.95$ and reaches $0.1856$ at $0.99$.

\begin{table}[h]
\caption{Hard-constrained CORAL frontier under the exact support restriction for the $p=18$ simulated system.}
\label{tab:support-frontier}
\centering
\scriptsize
\begin{tabular}{@{}rrrr@{}}
\toprule
$\rho_{\min}$ & Max $|\mathrm{offdiag}|$ & Mean $|\mathrm{offdiag}|$ & Min fidelity\\
\midrule
0.00 & 0.1622 & 0.0302 & 0.0000\\
0.30 & 0.1962 & 0.0310 & 0.3000\\
0.50 & 0.2001 & 0.0296 & 0.5000\\
0.70 & 0.1832 & 0.0308 & 0.7000\\
0.85 & 0.1624 & 0.0350 & 0.8500\\
0.95 & 0.1702 & 0.0483 & 0.9500\\
0.99 & 0.1856 & 0.0723 & 0.9900\\
\bottomrule
\end{tabular}
\end{table}

The non-monotonic maximum residual in Table~\ref{tab:support-frontier} is not inconsistent with the CORAL objective. CORAL minimizes aggregate squared residual correlation rather than the largest residual pair. A solution can therefore reduce the aggregate objective while allowing a larger individual residual correlation.

This distinction is explicit when the fidelity requirement is removed. Table~\ref{tab:support-floors} compares the best-found solution under CORAL's squared aggregate criterion with a separate min--max optimization. Minimizing $D_2$ gives maximum residual correlation $0.1839$ and mean absolute residual $0.0259$. Minimizing the worst pair reduces the maximum residual to $0.0729$ but increases the mean residual to $0.0605$. The two criteria therefore distribute residual dependence differently.

\begin{table}[h]
\caption{Best-found exact-support residuals with no fidelity requirement. The reported values are multistart numerical optima and are not globally certified minima.}
\label{tab:support-floors}
\centering
\scriptsize
\begin{tabular}{@{}lrrr@{}}
\toprule
Target criterion & Target value & Max $|\mathrm{offdiag}|$ & Mean $|\mathrm{offdiag}|$\\
\midrule
Squared aggregate $D_2$ & 0.2500 & 0.1839 & 0.0259\\
Worst pair $\delta_{\infty,\Omega}$ & 0.0729 & 0.0729 & 0.0605\\
\botrule
\end{tabular}
\end{table}

Proposition~\ref{prop:support-floor} establishes the condition under which a positive support-induced residual floor must exist, but the present non-convex computations do not prove that the realized support excludes every exact decorrelator. They also do not certify the global numerical values of the minima in Table~\ref{tab:support-floors}. The results should therefore be interpreted as best-found evidence that exact support materially changes the attainable decorrelation geometry, not as a globally certified support floor.

Sparse PCA provides a secondary benchmark. At 85 nonzero loadings, close in order of magnitude to the CORAL support, sparse PCA produces maximum and mean residual correlations of $0.5639$ and $0.1226$. (Appendix~\ref{app:sparse-pca} reports the corresponding source-fidelity values). Increasing sparsity reduces the loading count but returns residual correlation toward the untransformed baseline.

\begin{table}[h]
\caption{Sparse PCA residual correlation at varying loading sparsity for the $p=18$ simulated correlation structure. Components are rescaled to satisfy $t_j^\top Rt_j=1$.}
\label{tab:sparse-pca-ablation}
\centering
\scriptsize
\begin{tabular}{@{}rrrr@{}}
\toprule
$\alpha$ & Nonzero loadings & Max $|\mathrm{offdiag}|$ & Mean $|\mathrm{offdiag}|$\\
\midrule
0.1 & 85 & 0.5639 & 0.1226\\
0.5 & 20 & 0.3742 & 0.1089\\
1.0--16.0 & 18 & 0.3742 & 0.1082\\
\botrule
\end{tabular}
\end{table}

The sparse PCA comparison is not support-identical because sparse PCA controls loading sparsity rather than the specific allowed-entry set $\Omega$. It therefore serves only as a secondary sparse-component benchmark. The ablation's primary result is narrower: restricting allowable mixing changes the feasible geometry and can materially increase residual dependence even when high source fidelity remains attainable.

\section{Discussion}\label{sec:discussion}

CORAL reframes interpretable decorrelation as selection within a constrained family of transformations rather than post hoc interpretation of a fixed rotation. Proposition~\ref{prop:exact-decorrelators} shows that exact decorrelation is highly non-unique: the full family is $R^{-1/2}O(p)$. The relevant question is therefore not whether PCA decorrelates, which it does, but which exact decorrelator best preserves the source-variable correspondence required by the application. The threshold $\rho_\star(R)$ formalizes the strongest common source fidelity compatible with exact decorrelation.

The computed bounds make the observed frontier transitions interpretable. At $p=6$ and $p=18$, a $0.95$ fidelity floor coexists with exact decorrelation. At $p=50$, $\rho_\star(R)$ lies tightly around $0.95$, and the primal residual at that target is only $0.0035$. WDI and Wine have lower thresholds, near $0.75$ and $0.83$, so residual correlation is unavoidable by $\rho_{\min}=0.85$. CORAL does not remove the fidelity--decorrelation trade-off; it identifies where that trade-off becomes unavoidable.

ZCA whitening provides a structural benchmark rather than merely an empirical one. ZCA is the exact decorrelator $T=R^{-1/2}$, corresponding to $Q=I$ in the family $R^{-1/2}O(p)$, and its source fidelities are the diagonal elements of $R^{1/2}$. Proposition~\ref{prop:rho-bounds} therefore gives
$
\rho_\star(R)\geq\min_j (R^{1/2})_{jj},
$
where the right-hand side is exactly ZCA's minimum source fidelity. ZCA thus supplies a constructive lower bound on the maximum fidelity attainable under exact decorrelation. Equality occurs when $Q=I$ is max--min optimal; otherwise another orthogonal factor improves the weakest source match. This statement concerns $\rho_\star(R)$, not every primal CORAL solution: when $\rho_{\min}<\rho_\star(R)$, multiple exact decorrelators may satisfy the declared floor, and the primal objective does not require the selected solution to maximize fidelity beyond that floor.

The PCA comparison addresses a different question. Eigenvector signs are arbitrary, and principal components have no intrinsic one-to-one correspondence with original variables. The bottleneck and sum-optimal assignments therefore give PCA the most favorable sign-invariant assignment for the fidelity measure being reported. Even under this favorable benchmark, PCA's best minimum source fidelity remains well below the common fidelity available from constructed exact decorrelators in all five systems. This does not imply that CORAL dominates PCA for dimension reduction. PCA remains appropriate when variance concentration or low-dimensional representation is the objective; CORAL addresses settings in which transformed coordinates must remain individually anchored to designated source variables. Unlike ZCA, PCA's assignment criterion is unrelated to the max--min optimization defining $\rho_\star(R)$, so no analogous ordering follows a priori.

The support-restriction ablation identifies a separate limitation. Restricting which variables may be mixed changes the feasible geometry by intersecting the unit-variance ellipsoids with coordinate subspaces. Proposition~\ref{prop:support-floor} identifies the condition under which a positive residual floor must exist, while the numerical results show that its apparent magnitude depends on whether aggregate squared correlation or the worst residual pair is optimized. The residual plateau observed under the restricted support should therefore not be interpreted as a universal algebraic constant. Exact support and an explicit residual criterion are required before a numerical floor has a defensible interpretation. Sparse PCA provides a secondary sparsity benchmark, but its loading sparsity is not support-identical to CORAL's declared pattern and should not be interpreted as a direct optimization of the same feasible set.

Several limitations remain. The max--min $\rho_\star$ and support problems are non-convex. Achieved $\rho_\star$ values provide constructive lower bounds, while the analytical upper bounds are rigorous, but the numerical optimizer does not guarantee global optimality. The simulations use three dimensions and one moderate-correlation generator; the $p=50$ case is a scalability check rather than a controlled dimension-effect experiment. The support graph uses a fixed correlation threshold, and the empirical applications remain dense. Near-singular $R$ is a genuine boundary condition because high source fidelity across nearly redundant variables may conflict with decorrelation.

Future work should sharpen the bounds on $\rho_\star(R)$, develop scalable alternatives to exhaustive subset enumeration, derive sufficient conditions for support-compatible exact decorrelation, and seek stronger global optimality guarantees. The anchoring framework could also be generalized from designated source coordinates to domain-specified basis vectors encoding expected direction and approximate relative magnitude, with optional sign or interval constraints on selected transformation coefficients. Larger simulations should vary spectra and correlation topology systematically, while empirical support and basis structures should arise from domain knowledge or inferential graph construction. Finally, a complementary minimax formulation could replace the aggregate squared-correlation objective with the maximum absolute residual correlation, allowing direct control of the worst remaining pairwise dependence and supporting a broader Pareto analysis of source fidelity and decorrelation. 

\section{Conclusion}\label{sec:conclusion}

CORAL targets a simple requirement: decorrelate while preserving each transformed coordinate's source identity. All square exact decorrelators are $R^{-1/2}Q$, $Q\in O(p)$, which defines $\rho_\star(R)$, the maximum common source fidelity compatible with exact decorrelation. Certified intervals explain the frontiers across three simulations and two real datasets. At $p=50$, an exact decorrelator retains minimum source fidelity $0.949077$ versus PCA's best bottleneck fidelity $0.1716$. Loss of source-variable identity is therefore not inherent to exact decorrelation.

The method also clarifies what changes when mixing is restricted. Exact support creates a separate constrained geometry, and residual-correlation claims must distinguish the aggregate squared objective from worst-pair control. CORAL therefore contributes more than a new rotation algorithm: it provides a framework for stating, solving, and theoretically interpreting the trade between source-variable fidelity, residual dependence, and admissible mixing structure.

\backmatter

\bmhead{Acknowledgements}
The authors thank colleagues and reviewers who provide feedback on the development of the CORAL framework. 

\section*{Declarations}

\bmhead{Funding}
Not applicable.

\bmhead{Competing interests}
The authors declare no competing interests.

\bmhead{Ethics approval and consent to participate}
Not applicable.

\bmhead{Consent for publication}
Not applicable.

\bmhead{Data availability}
The World Development Indicators are publicly available from the World Bank. The Wine dataset is publicly available through the UCI Machine Learning Repository. The simulated correlation structures are generated algorithmically as described in Section~\ref{sec:simulation}.

\bmhead{Materials availability}
Not applicable.

\bmhead{Code availability}
The CORAL implementation, code used to reproduce the simulation and empirical analyses, and the interactive Streamlit application are publicly available at \url{https://github.com/dustoff06/CORAL}.

\bmhead{Use of generative AI}
During the preparation of this manuscript, the authors used ChatGPT 5.5 (OpenAI) and Claude Sonnet 5 (Anthropic) to assist with programming and prose editing. All methodological decisions, analyses, results, interpretations, and final manuscript content were reviewed and approved by the authors, who assume full responsibility for the work.

\bmhead{Author contributions}
L.V.F. conceived the study, developed the methodology, implemented the primary analyses, interpreted the results, and drafted the manuscript. C.F. reviewed the mathematical formulation and theoretical results. A.S. reviewed the computational implementation and reproducibility of the code. A.T. reviewed the manuscript for clarity, structure, and presentation. All authors reviewed and approved the final manuscript.

\bibliography{references}
\clearpage
\noindent
\Large{\textbf{Appendices}}

\appendix

\section{Final-run numerical verification}
\label{app:final-verification}
The final numerical run repeated the dense, support-restricted, and empirical analyses using the same implementations reported in the main text. Table~\ref{tab:final-diagnostics} summarizes numerical constraint satisfaction at the most demanding fidelity level, $\rho_{\min}=0.99$. All reported hard-constrained solutions satisfy the declared fidelity floor to the solver tolerance. The support-restricted solutions also satisfy the declared zero pattern exactly by construction.

\begin{table}[h]
\caption{Numerical diagnostics from the final hard-constrained run at $\rho_{\min}=0.99$.}
\label{tab:final-diagnostics}
\centering
\begin{tabular}{@{}lrrrr@{}}
\toprule
System & Min fidelity & Max $|\mathrm{offdiag}|$ & Mean $|\mathrm{offdiag}|$ & Max violation\\
\midrule
Synthetic $p=6$  & 0.9899 & 0.1514 & --     & $9.51\times10^{-5}$\\
Synthetic $p=18$ & 0.9900 & 0.2045 & 0.0535 & $3.99\times10^{-5}$\\
Synthetic $p=50$ & 0.9899 & 0.1700 & 0.0406 & $8.13\times10^{-5}$\\
WDI               & 0.9900 & 0.8296 & --     & $1.62\times10^{-5}$\\
Wine              & 0.9900 & 0.7605 & --     & $1.21\times10^{-5}$\\
\bottomrule
\end{tabular}
\end{table}

The optimized exact-decorrelation searches produced the constructive lower bounds $\rho_\star(R)\geq0.971817$, $0.958528$, and $0.949077$ for the $p=6$, $p=18$, and $p=50$ simulated systems, respectively. Combined with the corresponding analytical upper bounds, the resulting intervals are $[0.971817,0.975660]$, $[0.958528,0.961614]$, and $[0.949077,0.950352]$. For WDI and Wine, exhaustive subset enumeration gives the tighter intervals $[0.750852,0.757853]$ and $[0.826231,0.835487]$, respectively. These intervals use the attained max--min solution as a constructive lower bound and the analytical result of Proposition~\ref{prop:rho-bounds} as the upper bound; global optimality of the numerical max--min search is not assumed.

For the $p=50$ scalability check, the five-point hard-primal frontier required $5.42$ seconds, while the 100-start max--min search for $\rho_\star(R)$ required $130.55$ seconds. Exhaustive subset enumeration was not attempted at $p=50$ because it would require examination of $2^{50}-1$ nonempty subsets; the reported upper bound therefore uses the full-set trace bound.

\section{Sparse PCA comparison}
\label{app:sparse-pca}
Sparse PCA provides a secondary benchmark for the $p=18$ support-restriction analysis. The sparsity parameter was selected to approximate the declared CORAL support density. At $\alpha=0.1$, sparse PCA retained 85 nonzero loadings, compared with 110 free coefficients under the declared CORAL support. After rescaling each component to satisfy $t_j^\top Rt_j=1$ and applying the same sign-invariant source assignment used for the PCA benchmark, sparse PCA attained minimum source fidelity $0.4337$ and mean source fidelity $0.8343$. Its maximum and mean absolute residual correlations were $0.5639$ and $0.1226$, respectively. Because sparse PCA controls total loading sparsity rather than the specific allowed-entry pattern $\Omega$, this comparison is descriptive rather than support-identical.

\begin{table}[h]
\caption{Sparse PCA benchmark matched approximately to the declared support density of the $p=18$ CORAL ablation.}
\label{tab:sparse-pca-app}
\centering
\begin{tabular}{@{}rrrrrr@{}}
\toprule
$\alpha$ & Nonzero loadings & Min fidelity & Mean fidelity & Max $|\mathrm{offdiag}|$ & Mean $|\mathrm{offdiag}|$\\
\midrule
0.1 & 85 & 0.4337 & 0.8343 & 0.5639 & 0.1226\\
\bottomrule
\end{tabular}
\end{table}

\end{document}